\documentclass[reprint,twocolumn,superscriptaddress,showpacs,nofootinbib,notitlepage,preprintnumbers,aps,pra,floatfix]{revtex4-2}
\usepackage[utf8]{inputenc}
\usepackage{graphicx}
\usepackage{latexsym,amsmath,amssymb,lmodern,float,url}
\usepackage{qcircuit}
\usepackage{natbib}
\usepackage{color}
\usepackage{microtype}
\usepackage{bbold}
\usepackage{slashed}
\usepackage{multirow}
\usepackage{tikz}
\usepackage{xcolor}
\usetikzlibrary{shapes}
\usepackage{adjustbox}
\usepackage{todonotes}
\usepackage{braket}
\usepackage{algpseudocode}
\usepackage[linesnumbered,ruled,vlined]{algorithm2e}
 \usepackage[normalem]{ulem}
\usepackage{amsthm}
\usepackage{pgfmath}
\usepackage{todonotes}

\usepackage[colorlinks=true,backref=false, linktocpage=true,
citecolor=blue,urlcolor=blue,linkcolor=blue,pdfpagemode=UseOutlines]{hyperref}
\hypersetup{%
 bookmarksnumbered=true,
 pdftitle = {},
 pdfsubject = {},
 pdfauthor = {},
 pdfkeywords = {}
}
\usepackage{cleveref}
\Crefname{appendix}{App.}{Apps.}
\Crefname{equation}{Eq.}{Eqs.}
\Crefname{figure}{Fig.}{Figs.}

\newcommand{\Tr}{\operatorname{Tr}}
\newcommand{\ad}{\operatorname{ad}}
\renewcommand{\L}{\mathcal{L}}

\renewcommand{\d}{\mathrm{d}}

\newtheorem{lemma}{Lemma}

\newtheorem{definition}{Definition}

\begin{document}
 \title{State-space gradient descent and metastability in quantum systems}

 \author{Shuchen Zhu}
\email{shuchen.zhu@duke.edu} 
\affiliation{Department of Mathematics, Duke University, Durham, NC 27708, USA}

 \author{Yu Tong}
\email{yu.tong@duke.edu}
\affiliation{Department of Mathematics, Duke University, Durham, NC 27708, USA}
\affiliation{Department of Electrical and Computer Engineering, Duke University, Durham, NC 27708, USA}
\affiliation{Duke Quantum Center, Duke University, Durham, NC 27701, USA}

\date{\today}
\begin{abstract}
We propose a quantum algorithm, inspired by ADAPT-VQE, to variationally prepare the ground state of a quantum Hamiltonian, with the desirable property that if it fails to find the ground state, it still yields a physically meaningful local-minimum state that oftentimes corresponds to a metastable state of the quantum system. 
At each iteration, our algorithm reduces the energy using a set of local physical operations. The operations to perform are chosen using gradient and Hessian information that can be efficiently extracted from experiments.
We show that our algorithm does not suffer from the barren plateau problem, which is a significant issue in many variational quantum algorithms.
We use numerical simulation to demonstrate that our method reliably produces either the true ground state or a physically meaningful metastable state in typical physical systems with such states.
\end{abstract}
\maketitle

\section{Introduction}\label{sec:introduction}

Variational quantum algorithms (VQAs) are crucial for near-term quantum computing as they are well-suited for current noisy intermediate-scale quantum (NISQ) devices~\cite{Peruzzo2014,Cerezo2021}. VQAs require shallow circuits and utilize the power of classical optimization, making them practically implementable today for applications like quantum chemistry, optimization, and machine learning~\cite{Bharti2022,Cerezo2021}.
Despite their practical advantages, VQAs encounter several challenges that limit their performance. One critical issue is the barren plateau phenomenon, where gradients vanish exponentially with increasing qubit numbers, hindering the training of quantum circuits~\cite{McClean2018,Cerezo2021b,Larocca2025,cerezo2024doesprovableabsencebarren}.  Additionally, the complex energy landscapes, with numerous local minima, can trap optimization algorithms and prevent convergence to global minima~\cite{Bittel2021}. 
The numerous local minima are also a result of the circuit ansatz and may not correspond to any physically meaningful state~\cite{Ryabinkin2020,Bittel2021}.

The ADAPT-VQE \cite{GrimsleyEconomouBarnesMayhall2019adaptive, gustafson2024surrogateconstructedscalablecircuits,Grimsley2023,PhysRevResearch.6.013254} provides a promising approach to tackle the barren plateau problem. By gradually growing the quantum circuit through the addition of gates that best reduce the energy rather than using a fixed ansatz, ADAPT-VQE can oftentimes escape from ansatz-dependent local minima \cite{GrimsleyBarronBarnesEtAl2023}. 
If one views a parameterized quantum circuit as a mapping from the parameter space to the space of quantum states, which we will refer to as the \textit{state space} henceforth, then ADAPT-VQE can be seen as directly performing operations on the state space rather than on the parameter space. 
Another example where state-space information plays a role is the quantum natural gradient algorithm \cite{StokesIzaacKilloranEtAl2020quantum,KoczorBenjamin2022quantum,Patel2024natural}, which also alleviates the problem of local minima \cite{WierichsGoglinEtAl2020avoiding}.

In this work, we adopt this approach of operating in the state space to iteratively reduce the energy, thus steering the quantum state towards the ground state. This iteration will eventually reach a \textit{local-minimum state}, where a given set of operations can no longer reduce the energy. The ground state is naturally a valid local-minimum state, but there are likely other local-minimum states, and it is reasonable to ask whether these local-minimum states are physically meaningful or merely artifacts of algorithmic failure.

We observe that the set of local-minimum states depends on the available operations. 
If we are only allowed to use unitary operations, as in ADAPT-VQE, then the maximally mixed state is already a local-minimum state. 
This runs counter to physical intuition since it is easy to reduce the energy using non-unitary operations in such a scenario. 

Moreover, \cite{ChenHuangPreskillZhou2023local} proves that a Haar-random quantum state almost certainly has exponentially vanishing energy gradients with system size, making it practically indistinguishable from a local minimum.
To make the local-minimum states physically meaningful, we therefore need to implement non-unitary operations by introducing an ancilla register, which is one key difference between our approach and ADAPT-VQE.

Some of these ideas have been explored in connection to ADAPT-VQE. In particular, the state-space picture has been discussed in \cite{TangChenBiswasEtAl2024non} where the quantum state is updated through increasing circuit depth. Our method differs from \cite{TangChenBiswasEtAl2024non} mainly in that we use ancilla qubits, and we update the quantum state with a set of operators rather than just one, which resembles how gradient descent differs from coordinate descent in classical optimization. Ancilla qubits have also been used for ADAPT-VQE \cite{warren2022adaptive,sambasivam2025tepid} but mainly for Gibbs state preparation, which is a different task from what we are considering in this work.

Our algorithm, which we name \textit{state-space gradient descent} (SSGD), ensures that at convergence it yields a local-minimum state such that, if perturbed by an operation characterized by a local Lindbladian generator, the energy cannot be further decreased. We present numerical evidence that such a state is either the ground state or likely a physically meaningful \textit{metastable state}.

A metastable state is a long-lived physical configuration stable against small disturbances but ultimately transitioning to a lower-energy state \cite{2025_Yin,2016_Macieszczak}. 
These states arise in various quantum systems such as quantum dots, superconductors, trapped ions, and ultracold gases, enabling detailed observation of short-lived quantum phenomena \cite{2015_Mason,2019_Tanzi,2024_Pucher,2018_Earnest}. Metastable states are fundamental to quantum computing and memory applications, allowing stable manipulation and readout of quantum information \cite{2023_DeBry,2024_Pucher,2018_Earnest}. They also provide insights into non-equilibrium quantum dynamics like prethermalization, quantum scars, and exotic phases such as supersolids and false vacuum decay \cite{2025_Yin,2019_Tanzi,2024_Zenesini,2025_Vodeb}. 

Metastable states are oftentimes also local minima of the energy landscape: in order to escape from metastable states, one needs to pass through an energy barrier. In dynamics where energy-decreasing updates are favored, such as Glauber dynamics, the long lifetime can be seen as a result of the energy barrier since successive energy-increasing moves occur with very small probability. Recently Yin et al.~\cite{2025_Yin} proposed a mathematical theory for metastable states and rigorously proved this connection in the quantum setting. However, the perturbations they considered do not necessarily correspond to physically realizable completely positive trace-preserving maps (CPTP), and no algorithm to find such states was provided. Chen et al.~\cite{ChenHuangPreskillZhou2023local} studied local minima in quantum systems and provided an efficient algorithm for this task, but the algorithm involves a complicated procedure for quasi-local Lindbladian simulation, putting it beyond the reach of current devices. A recent concise review on efficiently locating local minima with quantum computers is provided in Ref.~\cite{Li2025review}.

In this work, we consider local-minimum states that are stable to local physically realizable perturbations, and propose an algorithm to find such states on near-term devices. Although we consider stability under perturbations generated by Lindbladians, our algorithm does not rely on any Lindbladian simulation algorithm. We numerically simulate our algorithm for the 1D transverse-field Ising model and neutral atoms arranged on a ring, which are typical quantum systems known to exhibit metastable states, and observe that it consistently converges to either the ground state or the metastable state. Our algorithm is provably free from the barren plateau issue.

\section{Theory}

The SSGD algorithm that we are going to describe iteratively updates the quantum state by local operations to reduce the energy. It uses an ancilla register, which in the simplest case consists of one ancilla qubit, to implement non-unitary operations. We denote the original system consisting of $N$ qubits by $S$ and the ancilla register by $A$.

The local operations we consider are generated by a set of generators (typically Pauli operators) $\mathcal{G}=\mathcal{G}_A\cup \mathcal{G}_S$, where generators in $\mathcal{G}_A$ act jointly on the ancilla $A$ and system $S$, and generators in $\mathcal{G}_S$ act only on $S$. We explicitly write out $\mathcal{G}=\{P_1,P_2,\cdots,P_d\}$, and denote $\vec{P}=(P_1, P_2, \dots, P_d)$. For concreteness we consider a single-qubit ancilla register, let $\mathcal{G}_A$ be the set of all Pauli operators that act on the ancilla qubit and at most $k-1$ adjacent systems qubits, and $\mathcal{G}_S$ be the set of all Pauli operators that act on at most $k$ adjacent system qubits. We will always use this setup unless otherwise stated, even though our algorithm works for more general choices of $\mathcal{G}_A$ and $\mathcal{G}_S$ as well. More ancilla qubits are helpful for parallelizing operations but otherwise do not change the performance of the algorithm.

Our goal is to get a local-minimum state, i.e., a state for which local operations cannot further decrease the energy without passing through an energy barrier. Obviously, the ground state is a local minimum, while other states satisfying this criterion have a close connection to metastable states, as will be discussed in the numerical results section.

At each iteration, we start with a state $\rho$ in $S$, and the ancilla register is initialized in the $\ket{0}$ state. Therefore the joint state of $AS$ is $\tilde{\rho}=\ket{0}\bra{0}\otimes\rho$.
We then update $\tilde{\rho}$ using these Pauli operators as follows. For $\vec{\theta} \in \mathbb{R}^d$,  $\vec{\theta} \cdot \vec{P} = \theta_1 P_1 + \dots + \theta_d P_d$.
The unitary transformation is given by  
\begin{align}
 \tilde\rho(\vec\theta) = U(\vec{\theta}) \tilde\rho U^\dagger(\vec{\theta}),
\end{align}
where $U(\vec{\theta}) = e^{-i\vec{\theta} \cdot \vec{P}}$,
with  energy   $E(\vec{\theta}) = \Tr(\tilde\rho(\vec{\theta}) \tilde H)$ where $\tilde{H}=I\otimes H$. We finally reset the ancilla register, so that the system register is in the state $\Tr_A(\tilde\rho(\vec\theta))$.

Similar to local minima in classical optimization, for the local-minimum states we are looking for, the first-order optimality condition needs to be satisfied:  
\begin{align}
\left.\frac{\partial}{\partial\theta_j}E(\vec\theta)\right|_{\vec{\theta}=0}=-i\Tr([P_j,\tilde\rho]\tilde H) =0,\label{eq: first-order}
\end{align}
which is equivalent to $\Tr(\ad_{P_j}(\tilde\rho)\tilde H)=0$, where $\ad_{A}(B):=[A,B]$ for operators $A$ and $B$.
A second-order condition is also necessary to distinguish local minima from saddle points. We first define the Hessian matrix $K=(K_{jk})$ where
\begin{align}\label{eq: Hessian_def}
    K_{jk} = \left.\frac{\partial^2}{\partial\theta_j\partial\theta_k}E(\vec\theta)\right|_{\vec\theta=0}
    =  -\frac{1}{2}\Tr(\{\ad_{P_k},\ad_{P_j}\}(\tilde\rho)\tilde H),
\end{align}
and the second-order condition is $K = (K_{jk})\succeq 0$. 

One may also consider the optimality with respect to Lindbladian perturbation. We may perturb the quantum state using a CPTP map $e^{i\theta \mathcal{L}}$ for small $\theta$, where $\mathcal{L}(\cdot)=L\cdot L^\dag-\frac{1}{2}\{L^\dag L,\cdot\}$ is a Lindbladian generator on $k-1$ qubits. The optimality condition is then
\begin{equation}
    \label{eq:lindblad_optimality}
    \frac{\d}{\d \theta}\Tr[H e^{\theta\mathcal{L}}(\rho)]\geq 0.
\end{equation}
Note that we only need to consider $\theta\geq 0$ otherwise $e^{\theta\mathcal{L}}$ would not be CPTP.
Surprisingly, this Lindbladian optimality condition is implied by the second-order condition \eqref{eq: Hessian_def}, as proved in Lemma~\ref{lem:second_order_condition_implies_lindblad_optimality}. We therefore do not need to implement Lindbladian evolution in our algorithm but only need to apply unitaries to satisfy \eqref{eq: first-order} and \eqref{eq: Hessian_def}.

With the above optimality conditions, we can give a formal definition of a local-minimum state
\begin{definition}
    \label{defn:local_minimum_state}
    A state $\rho$ is a local-minimum state with respect to perturbations generated by $\mathcal{G}=\mathcal{G}_A\cup\mathcal{G}_S$ if it satisfies \eqref{eq: first-order} and its corresponding Hessian matrix defined \eqref{eq: Hessian_def} is positive semi-definite.
\end{definition}

Our algorithm converges to a state that approximately satisfies the above conditions. Even though the lifetime aspect of the metastable state is not captured in the above definition, all the local-minimum states we find in numerical experiments are long-lived metastable states, giving evidence of the close connection between these two concepts.

\subsection{The effect of ancilla and barren plateau}

The introduction of an ancilla qubit is important for avoiding the barren plateau issue that plagues many variational quantum algorithms. If we only allow unitary operations generated by a set $\mathcal{G}$ of $k$-local generators, then by \cite[Lemma~C.1]{ChenHuangPreskillZhou2023local}, for any $\epsilon>2^{-n/4}$, with probability at least $1-2^{-2^{n/4}}$, a Haar-random state $\ket{\psi}$ approximately satisfies the first-order optimality condition in the sense that
\[
\left|\frac{\d}{\d\theta} \Tr[e^{i\theta P}\ket{\psi}\bra{\psi}e^{-i\theta P}H]\Big|_{\theta=0}\right|\leq \epsilon,\quad \forall P\in\mathcal{G}.
\]
This means that without the ancilla qubit, there will be a multitude of physically irrelevant local minima, making it difficult for the algorithm to either reach the ground state or any physically meaningful metastable states.

We will next discuss a particular modification of our algorithm that is provably free from the barren plateau problem. For simplicity, we consider a 1D circuit even though the result can be straightforwardly generalized to higher dimensions. 
We index the system qubits by $1,2,\cdots,N$, and use $N$ ancilla qubits similarly indexed.
We set $\mathcal{G}_A$ to be the set of Pauli operators supported on the $j$th system qubit and the $j$th ancilla qubit for each $j$ (by this we include Pauli operators that act only on one qubit). $\mathcal{G}_S$ is chosen among 2 choices for each step. We let $\mathcal{G}_S$ include either all Pauli operators supported on the $(2j-1)$th and $2j$th qubits, $j=1,2,\cdots,N/2$, or those supported on the $2j$th and $2j+1$th qubits. 
Different from our original algorithm, now in each iteration we rotate among $\mathcal{G}_A$ and the two choices of $\mathcal{G}_S$, and only update the state with generators from the given set. This results in a brickwall circuit structure, as shown in Figure~\ref{fig:SSGD_circuit}.

The barren plateau issue manifests when the energy variance vanishes with increasing system size under randomly chosen parameters. In our setting, all parameters $\vec{\theta}_S$ and $\vec{\theta}_A$ are chosen randomly, such that each two-qubit unitary is effectively sampled from the Haar measure over $\mathrm{SU}(2)$. Up to any fixed iteration, our circuit becomes a special case of the dynamically parameterized circuit considered in \cite{DeshpandeHinscheEtAl2024dynamic} (more precisely their Definition~3). 

For a $k$-local Hamiltonian, we can directly apply~\cite[Theorem~1]{DeshpandeHinscheEtAl2024dynamic} to compute a lower bound on the energy variance:
\[
\mathrm{Var}_{\vec{\theta}}[\Tr(H\rho(\vec{\theta}))] \geq \frac{\|H\|_{\mathrm{HS}}^2}{5^{k(f+1)}},
\]
where $\|\cdot\|_{\mathrm{HS}}$ denotes the normalized Hilbert-Schmidt norm, and $f$ is the minimum distance between an observable and the nearest feedforward operation within its lightcone. 
Since an observable is separated by at most 5 circuit layers from a reset operation on the ancilla qubits within its lightcone, we have $f\leq 5$. Consequently, the variance is lower bounded by $\|H\|_{\mathrm{HS}}^2$ up to a constant factor, which does not vanish as $n$ increases.

\begin{figure}
    \centering
    \includegraphics[width=0.8\linewidth]{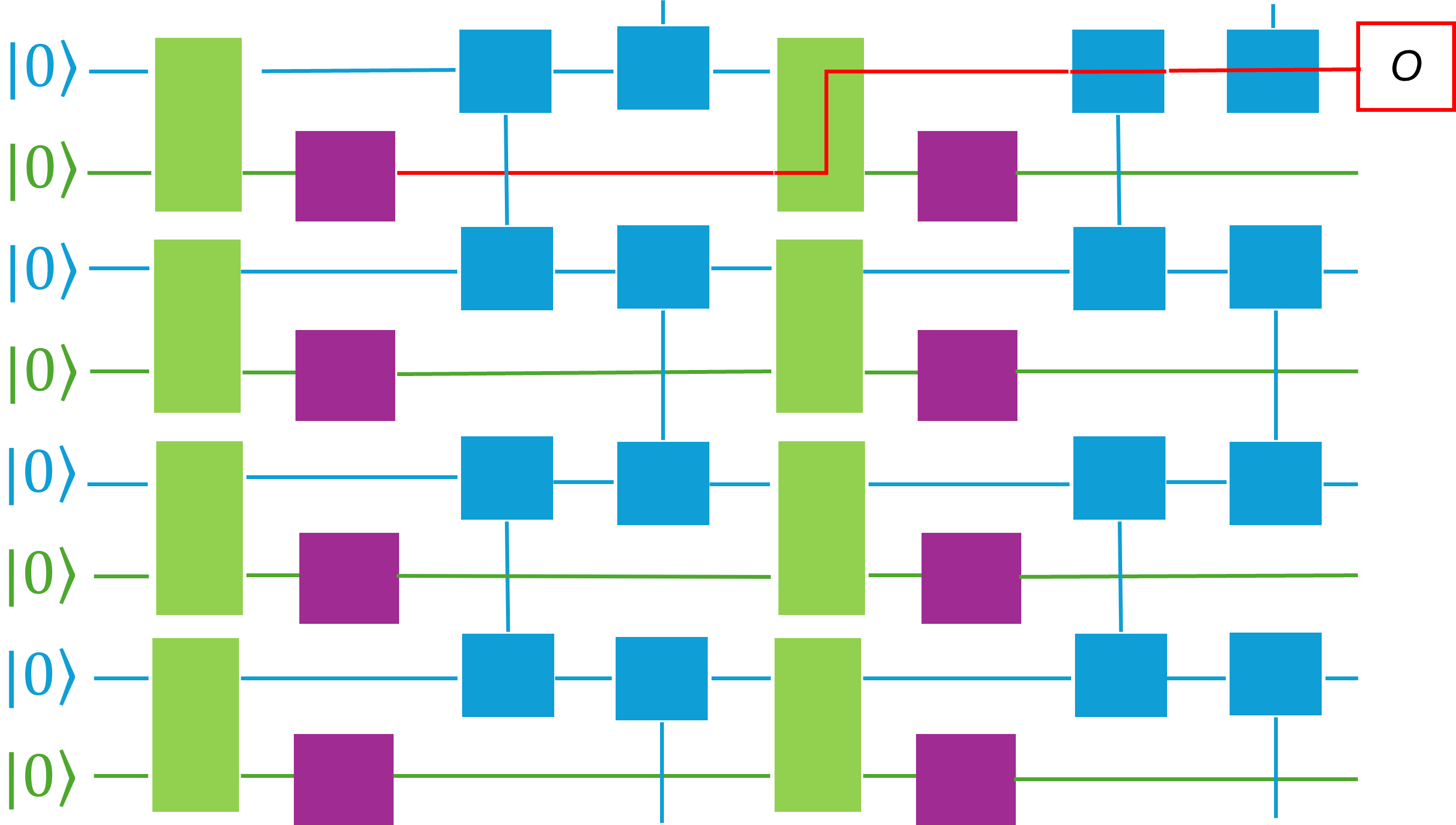}
    \caption{The circuit structure for SSGD. System qubits and ancilla qubits are colored blue and green respectively. Blue and green boxes represent two-qubit gates on system qubits and between system and ancilla qubits respectively. Purple boxes represent measurement and reset. The red line tracks how a single-qubit observable $O$ is connected to the nearest reset operation.}
    \label{fig:SSGD_circuit}
\end{figure}

\subsection{The algorithm}
%
We will outline our SSGD algorithm given a set of local generators $\mathcal{G}=\mathcal{G}_S\cup\mathcal{G}_A$.
Our algorithm works for any general choice of $\mathcal{G}$.
The pseudocode is outlined in Algorithm~\ref{alg:general_SGD}. 

\begin{algorithm}
\caption{General scheme of an SSGD algorithm.}
\label{alg:general_SGD}
 \KwData{Hamiltonian $H$, generators $\mathcal{G}=\mathcal{G}_A\cup\mathcal{G}_S$, initial state $\rho$, max iterations $T$, step sizes $\delta t_S, \delta t_A$.}
 \KwResult{A local-minimum state $\rho$ of $H$.}

\For{$t=1,2,\cdots,T$}{
    $\vec{g}_S\gets$ $\textsc{SystemDirection}(H, \rho, \mathcal{G}_S)$

    $\vec{g}_A\gets$ $\textsc{AncillaDirection}(H, \rho, \mathcal{G}_A)$

    $\vec{\theta}_S\gets \vec{g}_S\delta t_S$, $\vec{\theta}_A\gets \vec{g}_A\delta t_A$, $\vec{\theta}\gets(\vec{\theta}_A,\vec{\theta}_S)$.

    $U(\vec{\theta})\gets e^{-i(\vec{\theta}_S\cdot \mathcal{\vec G}_S + \vec{\theta}_A\cdot \mathcal{\vec G}_A)}$, 
    $\tilde\rho \gets U\tilde\rho U^\dagger$.

    Reset the ancilla register $A$.

}

 \textbf{return} $\tilde\rho$
\end{algorithm}

In the above $\vec{\mathcal{G}}_A$ and $\vec{\mathcal{G}}_S$ are the vectors containing elements of $\mathcal{G}_A$ and $\mathcal{G}_S$ respectively.
The subroutine \textsc{SystemDirection} computes the gradient of $E(\vec{\theta})$ corresponding to entries of $\vec{\theta}_S$. By \Cref{eq: first-order} the $j$th entry of this gradient is computed as
\begin{align}
(\vec{g}_S)_j 
&=-i \Tr\left( [(\mathcal{\vec{G}}_S)_j,\tilde{H}]  \tilde{\rho}\right) + \delta_j
\end{align}
where $\delta_j\sim \mathcal{N}(0, \sigma^2_j)$ is Gaussian noise that accounts for the inevitable statistical noise coming from measuring the observables $-i[(\mathcal{\vec{G}}_S)_j,\tilde{H}]$. The variance can be controlled by controlling the number of samples, and we choose $\sigma^2_j=\delta t_S$ in numerical experiments.

A different approach is needed for the generators acting on the ancilla qubits. Consider a Pauli generator $P=P_A\otimes P_S\in \mathcal{G}_A$. If $\braket{0|P_A|0}=\pm 1$, then $e^{-i \theta P}$ does not create any entanglement between $A$ (which starts in the state $\ket{0}$) and $S$, and one might as well just apply a lower-weight operator $P_S$ as an element in $\mathcal{G}_S$. Therefore we only need to include into $\mathcal{G}_S$ those $P=P_A\otimes P_S$ such that $\braket{0|P_A|0}=0$. For these operators, we can easily verify that the corresponding energy partial derivatives are $0$, and therefore the first-order conditions are already satisfied. We then need to rely on second-order information in order to make use of these generators.

The subroutine \textsc{AncillaDirection} therefore uses the Hessian matrix
$K_{jk} = -\frac{1}{2} \Tr\left( \{ \operatorname{ad}_{P_k}, \operatorname{ad}_{P_j} \} (\tilde{\rho}) \tilde{H} \right)$, for $P_j, P_k \in \mathcal{G}_A$.
Let $Q$ be the matrix of eigenvectors and $\vec{E}$ the vector of eigenvalues of $K$. Define the clipped eigenvalue vector $\vec{E}'$ by
\begin{align}
\vec{E}'_i = 
\begin{cases}
\vec{E}_i, & \text{if } \vec{E}_i < -E_{\mathrm{tol}}, \\
0,        & \text{otherwise}.
\end{cases}
\end{align}
where $E_{\mathrm{tol}}>0$ is a small number chosen to make the procedure robust against statistical noise on $K_{jk}$.
We choose the direction to move to be $\vec{g}_A = Q \vec{E}'$, which ensures that we move in a direction that reduces the energy according to the Hessian matrix, and directions along which the energy decreases more steeply are favored. This choice of $\vec{g}_A$ is not unique and one can explore many other possibilities.

\section{Numerical simulation}

We will apply our algorithm to the one-dimensional Transverse Field Ising Model (TFIM) and an antiferromagnetic neutral atom chain through numerical simulation. The simulation is done with the software package QuTiP \cite{johansson2012qutip}.

\subsection{One-dimensional Transverse Field Ising Model }

The TFIM is one of the most well-studied models in condensed matter physics~\cite{sachdev2011quantum,pfeuty1970one}.
It serves as a prototypical example of a quantum system that undergoes a quantum phase transition at zero temperature~\cite{sachdev2011quantum}. 
In its standard form, the Hamiltonian is given by:
\begin{equation}
    H = -J \sum_{i} \sigma_i^z \sigma_{i+1}^z - h_x \sum_{i} \sigma_i^x - h_z \sum_{i} \sigma_i^z,
    \label{eq:tfim}
\end{equation}
where $\sigma_i^x$ and $\sigma_i^z$ are the Pauli matrices acting on site $i$, $J$ is the nearest-neighbor interaction strength, and  $h_x, h_z$  represent the transverse and longitudinal field strengths, respectively.

In the absence of a longitudinal field $h_z = 0$, the system exhibits a well-known quantum phase transition at the critical point $h_x/J = 1$~\cite{pfeuty1970one}. 
Introducing a longitudinal field $h_z \neq 0$  explicitly breaks the $\mathbb{Z}_2$ symmetry of the standard TFIM~\cite{kormos2017real,liu2019confined}.
This perturbation makes the system non-integrable, leading to rich quantum many-body dynamics, including domain-wall confinement, dynamical oscillations, and slow thermalization~\cite{kormos2017real,liu2019confined,vovrosh2021confinement,mazza2019suppression}.

We perform numerical simulations of the SSGD dynamics for initial states of the 1D TFIM, incorporating both unitary and non-unitary gate updates while varying the ratio $h_x/J$ for small $h_z$.
A metastable state~\cite{lagnese2021false} can be prepared in certain regimes through quantum quenching. Initially, the system is prepared in a ferromagnetic state with all spins aligned in $\sigma^z_i = 1$ direction. The state is then evolved until it sufficiently converges to the ground state for $h_z < 0$. At this point, a quench is performed by switching $h_z \rightarrow -h_z$, causing the original ground state to transform into a metastable state.

\begin{figure}[!ht]
   \includegraphics[width=1\linewidth]{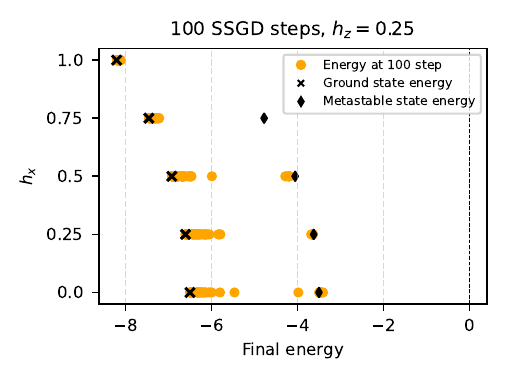}
    \caption{The final state energies of all 64 initial computational-basis states after 100 SSGD steps is shown, together with ground state and metastable state energies. 
    We observe that the final energies cluster around ground state and metastable state energies.
    The parameters used in the simulation are $J = 1$, $h_z = 0.25$, and $0 \leq h_x \leq 1$.
}
    \label{fig:TFIM_overall}
\end{figure}

For certain initial conditions, intertwining dissipative non-unitary evolution with unitary stochastic gradient descent (SGD) can accelerate the escape from the metastable state as shown in \Cref{fig:TFIM_1v1_000111_hx0}. 
However, the dissipative method does not always guarantee an advantage. For certain initial states, it may exhibit slower convergence at a late stage compared to purely unitary dynamics with $\mathcal{G}_A=\emptyset$, as illustrated in \Cref{fig:hx025_TFIM_111001.png}.

\begin{figure}[!ht]
   \includegraphics[width=1\linewidth]{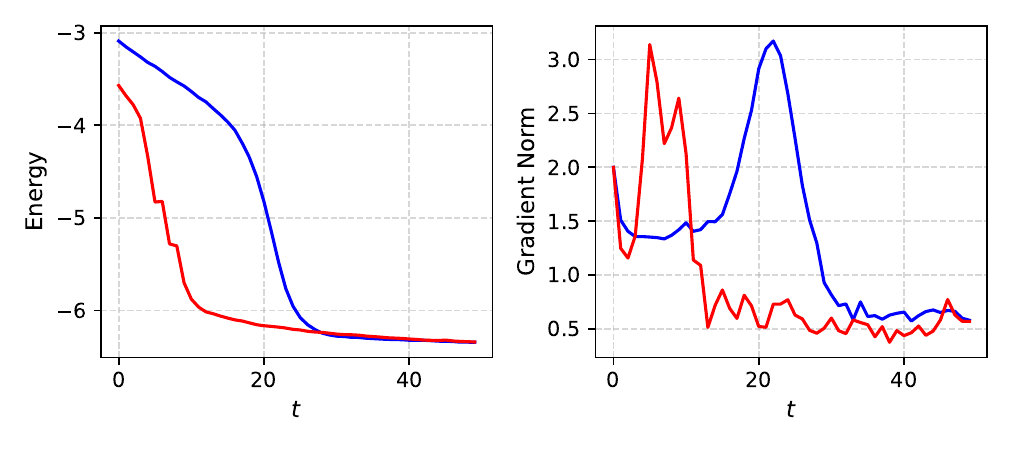}
    \caption{For the initial state $\ket{000111}$, we compare SSGD with dissipative evolution (red line) with SSGD with purely unitary evolution (blue line).  The parameters used are $J = 1$, $h_z = 0.25$, and $h_x = 0.25$.}
    \label{fig:TFIM_1v1_000111_hx0}
\end{figure}

\begin{figure}[!ht]
   \includegraphics[width=1\linewidth]{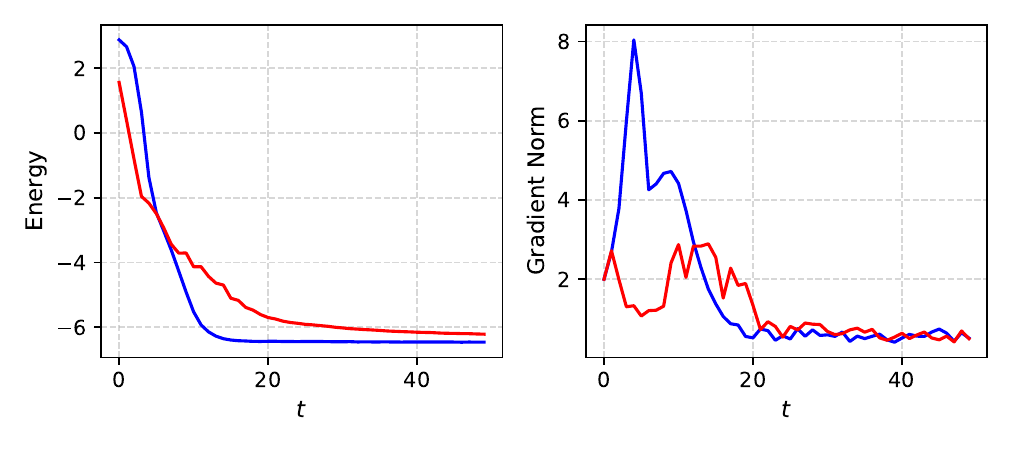}
    \caption{For the initial state $\ket{111001}$,  we compare SSGD with dissipative evolution (red line) with SSGD with purely unitary evolution (blue line). The parameters used are $J = 1$, $h_z = 0.25$, and $h_x = 0.25$.}
    \label{fig:hx025_TFIM_111001.png}
\end{figure}

\subsection{Rydberg atom}

Rydberg atom arrays, particularly one-dimensional chains of neutral atoms excited to Rydberg states, have attracted significant interest as controllable quantum systems with tunable long-range interactions, ideal for studying quantum phase transitions and non-equilibrium many-body phenomena~\cite{browaeys2020many, saffman2010quantum, bernien2017probing}.
Recent studies show that Rydberg atom arrays can achieve high-fidelity quantum gates, such as CZ and Toffoli, suitable for gate-based quantum computing~\cite{Levine2019,Evered2023}. Experiments have successfully demonstrated quantum algorithms including QAOA and quantum phase estimation~\cite{Graham2022}. Additionally, specialized error-correction approaches promise scalable, fault-tolerant architectures tailored to Rydberg-based quantum computers~\cite{Cong2022}.

A system of $N$ atoms arranged in a one-dimensional periodic chain is described by the Hamiltonian \begin{equation} H = \sum_{i=0}^{N-1} \left( \frac{\Omega}{2} \sigma_i^x - \Delta_i n_i \right) + \sum_{i<j} V_{ij} n_i n_j, \label{eq:Rydberg} \end{equation} where $\sigma_i^x$ represents the Rabi oscillations between the ground and Rydberg states of atom $i$, $\Omega$ is the Rabi frequency, $\Delta_i$ is the detuning, and $n_i = (1 - \sigma_i^z)/2$ is the occupation number operator. The interaction term $V_{ij} = C_6 / |r_i - r_j|^6$ represents the van der Waals interaction, giving rise to the Rydberg blockade effect, which prevents the simultaneous excitation of nearby atoms. We will use $\ket{0}$ to represent the atomic ground state and $\ket{1}$ to represent the Rydberg state.

In this Hamiltonian, in the absence of the $\sigma_i^x$ term and with appropriate values for the parameters, the two configurations with the lowest energies are $\ket{101010}$ and $\ket{010101}$, with the former being the ground state while the latter being a metastable state. With a non-zero $\sigma_i^x$ term we can identify the metastable state by the dominant computational basis state in $\rho$: the state with $\ket{101010}$ being dominant has lower energy and is the ground state, while $\ket{010101}$ indicates the metastable state.

For numerical simulation, we assume periodic boundary conditions and set the parameters as in Ref.~\cite{Darbha:2024srr} as $\Omega/\pi \in [1,5]$ MHz, $R_b=9.76$ $\mu$m, $a=8$ $\mu$m, $\Delta_{\text{glob}}/2\pi=2.5$ MHz, $\Delta_{\text{loc}}/2\pi=0.625$ MHz, where $\Delta_j := \Delta_{\text{glob}}+(-1)^j\Delta_{\text{loc}}$, and $R_b=(C_6/\Omega)^{1/6}$. The existence of metastable states is demonstrated in \Cref{fig:rydberg_overall} under $100$ gradient descent steps.
The metastable state in the Rydberg system is initialized from the $\mathbb{Z}_2$ state $\ket{010101}$, which approximates the false vacuum state. During evolution, the N{\'e}el order parameter $N_e=\frac{1}{N} \sum_j (-1)^j \sigma_z^{j}$ remains close to $1$ for $t > 0$, indicating that the system remains in the metastable state for a significant period. The decay of this state follows an exponential trend, characteristic of quantum tunneling rather than a phase transition.

\begin{figure}[]
   \includegraphics[width=1\linewidth]{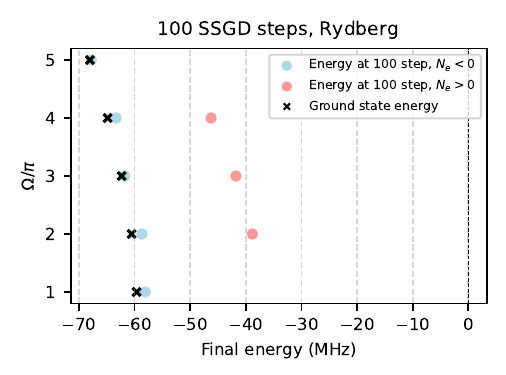}
    \caption{The final state energies of some initial computational-basis states after 100 SSGD steps is shown, together with the ground state energy. 
    The data points are colored according to the signs of their N{\'e}el order parameter values $\braket{N_e}$, with $\braket{N_e}>0$ corresponding to the metastable state and $\braket{N_e}<0$ corresponding to the ground state.
}
    \label{fig:rydberg_overall}
\end{figure}

We found that the dissipative evolution does not always provide an advantage over the purely unitary SSGD and sometimes exhibits slower convergence. However, when the unitary dynamics is initially trapped in a metastable state for an extended period, the dissipative method helps it escape more quickly, as shown in \Cref{fig:rydberg_omega2pi_010111}. While for some initial conditions, ancilla method does not gain any advantage as shown in \Cref{fig:rydberg_omega2pi_011100}.

\begin{figure}[]
   \includegraphics[width=1\linewidth]{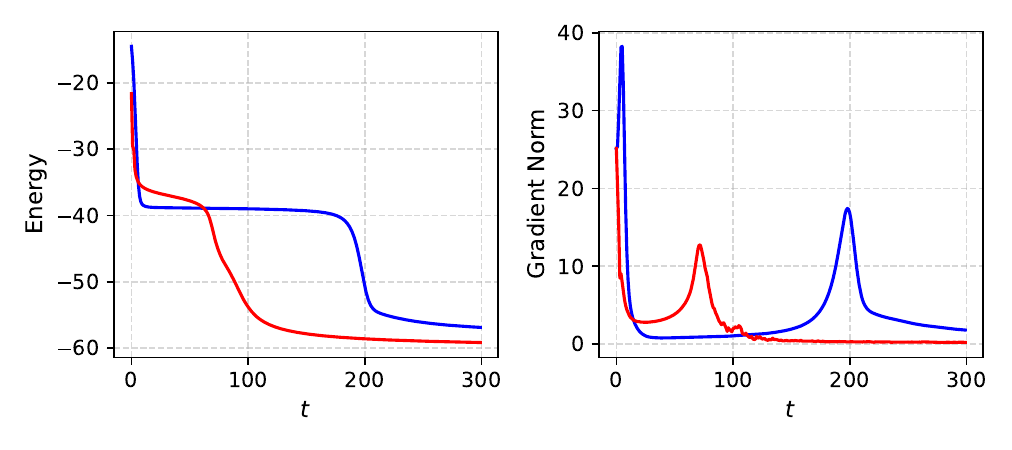}
    \caption{For the initial state $\ket{010111}$, we compare SSGD with dissipative evolution (red line) with SSGD with purely unitary evolution (blue line). $\Omega/2\pi=1.0$ MHz. }
    \label{fig:rydberg_omega2pi_010111}
\end{figure}

\begin{figure}[]
   \includegraphics[width=1\linewidth]{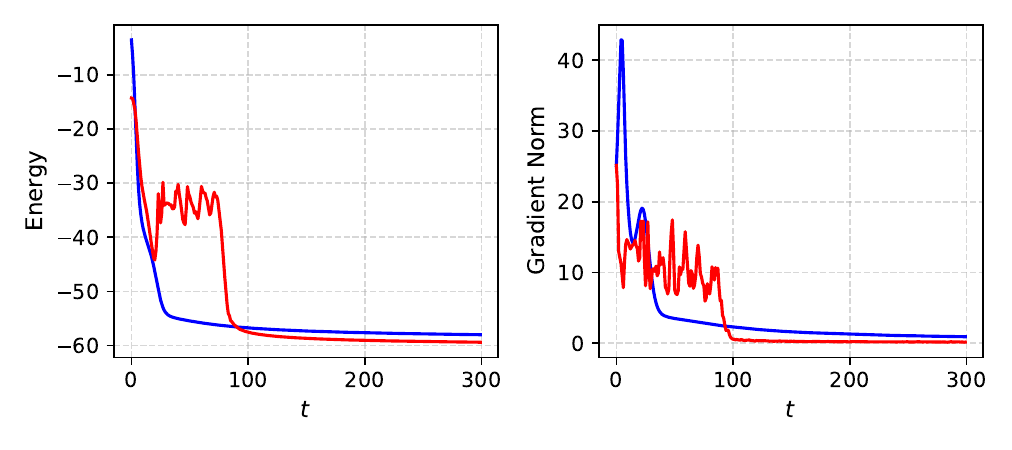}
    \caption{For the initial state $\ket{011100}$, we compare SSGD with dissipative evolution (red line) with SSGD with purely unitary evolution (blue line). $\Omega/2\pi=1.0$ MHz. }
    \label{fig:rydberg_omega2pi_011100}
\end{figure}

\section{Conclusion}

In this work we propose a variational quantum algorithm that, at convergence, prepares either the ground state or a local-minimum state, which we numerically find to be closely associated with long-lived metastable states in the 1D TFIM and a 1D antiferromagnetic neutral atom chain. This quantum algorithm is friendly for implementation on near-term devices, and provably does not suffer from the barren plateau issue.

Our numerical results give evidence to the connection between the local-minimum states we find and long-lived metastable states. It is of interest to establish a mathematically rigorous connection between the two, similar to what was done in \cite{2025_Yin}. To do this we would need to prove that the energy barrier present in a local-minimum state leads to long lifetime. We may also further explore the connection between the local-minimum states with the Lindbladian mixing time \cite{2016_Macieszczak,rouze2024efficient,fang2025mixing,tong2024fast,smid2025polynomial,zhan2025rapid}.

It is also natural to consider an experimental implementation of our algorithm. While the algorithmic framework is simple and does not use any complex quantum algorithm component, we may need certain modifications to make the most of current devices with limited system size and circuit depth. We may incorporate elements from classical optimization algorithm, such as line search, or from other variational quantum algorithms such as ADAPT-VQE, to achieve faster convergence with lower circuit depth.

While this work primarily proposes a quantum algorithm, this framework is flexible enough to yield a quantum-inspired algorithm for finding metastable states on classical computers, using methods such as tensor networks and neural-network quantum states \cite{carleo2017solving}.

\begin{acknowledgements}
The authors thank Edwin Barnes and Di Fang for helpful discussion.
SZ acknowledges support from the U.S. Department of Energy, Office of Science, Accelerated Research in Quantum Computing Centers, Quantum Utility through Advanced Computational Quantum Algorithms, grant No. DE-SC0025572.
\end{acknowledgements}

\bibliography{bibo}

\appendix 

\section{Technical lemmas}
\begin{lemma}\label{lem: Hessian_form}
    The Hessian $(K)_{ij}$ defined in \Cref{eq: Hessian_def} equals to $-\frac{1}{2}\Tr[\{\ad_{P_k},\ad_{P_j}\}(\tilde\rho)\tilde H]$.
\end{lemma}
\begin{proof}
    \begin{align}
    K_{jk} 
    &= \left.\frac{\partial^2}{\partial\theta_j\partial\theta_k} 
    \Tr\big( U(\vec{\theta})\, \tilde{\rho}\, U^\dagger(\vec{\theta})\, \tilde{H} \big) 
    \right|_{\vec{\theta}=0} \notag \\
    &= \left. \frac{\partial}{\partial \theta_j} 
    \Tr\big( (\partial_k U)\, \tilde{\rho}\, U^\dagger\, \tilde{H} + 
             U\, \tilde{\rho}\, (\partial_k U^\dagger)\, \tilde{H} \big) 
    \right|_{\vec{\theta}=0} \notag \\
    &= \Tr\big( (\partial_j \partial_k U)\, \tilde{\rho}\, U^\dagger\, \tilde{H} + 
                (\partial_k U)\, \tilde{\rho}\, (\partial_j U^\dagger)\, \tilde{H} 
                \notag \\
    &\quad + (\partial_j U)\, \tilde{\rho}\, (\partial_k U^\dagger)\, \tilde{H} + 
              U\, \tilde{\rho}\, (\partial_j \partial_k U^\dagger)\, \tilde{H} 
    \big)_{\vec{\theta}=0} \notag \\
    &= -\frac{1}{2} \Tr\big( (P_j P_k + P_k P_j)\, \tilde{\rho}\, \tilde{H} 
            - 2 P_k\, \tilde{\rho}\, P_j\, \tilde{H} \notag \\
    &\quad - 2 P_j\, \tilde{\rho}\, P_k\, \tilde{H} + 
              \tilde{\rho}\, (P_j P_k + P_k P_j)\, \tilde{H} \big) \notag \\
    &= -\frac{1}{2} \Tr\big( \{ \operatorname{ad}_{P_k}, \operatorname{ad}_{P_j} \} (\tilde{\rho})\, \tilde{H} \big),
\end{align}
where the second last equality used the following identities:
\begin{align}
    \partial_k U|_{\theta=0} &= -iP_k, \quad \partial_k U^\dagger|_{\theta=0} = iP_k,\notag\\
    \partial_j\partial_k U|_{\theta=0} &= \partial_j\partial_k U^\dagger|_{\theta=0}= -\frac{1}{2}(P_jP_k+P_kP_j).
\end{align}
\end{proof}

The next lemma shows that the second-order condition implies first-order Lindbladian optimality.
\begin{lemma}
\label{lem:second_order_condition_implies_lindblad_optimality}
    We assume for any Pauli operator $P$ supported on $k-1$ adjacent qubits in $S$, $X\otimes P$ and $Y\otimes P$ are both contained in $\mathcal{G}_A$ where Pauli operators $X$ and $Y$ act on any of the ancilla qubit in $A$. 
    Let $\L(\rho)=L\rho L^\dagger - \frac{1}{2}\{L^\dagger L,\rho\}$, $L$ acts on $k-1$ adjacent qubits. Then
    \begin{align}
    \left.\frac{d}{dt}\Tr[e^{t\L}(\rho) H]\right|_{t=0}\geq 0,
    \end{align}
    if the second-order condition $K\succeq 0$ is satisfied.
\end{lemma}

\begin{proof}
 Write $L=A+iB$, where $A, B$ are Hermitian over $k-1$ adjacent qubits, then construct a unitary operator $G$ on one ancilla qubit and $k-1$ adjacent system qubits
\begin{align}
    G &= X\otimes A + Y\otimes B \notag\\
    &= \begin{pmatrix}
      0 & L^\dagger \\
      L & 0
    \end{pmatrix}
    =\sum_{j=1}^d\alpha_jP_j.
\end{align}
We first want to show $\Tr(\L(\rho)H) = -\frac{1}{2}\Tr[\operatorname{ad}_G^2(\tilde\rho)\tilde H]$. Note that 
\begin{align}
    \operatorname{ad}_G(\tilde\rho) = \begin{pmatrix}
0 & -\rho L^\dagger \\
L \rho & 0
\end{pmatrix},
\end{align}
and
\begin{align}
     \operatorname{ad}_G^2(\tilde\rho) = \begin{pmatrix}
\{L^\dagger L,\rho\} & 0 \\
0 & -2L\rho L^\dagger
\end{pmatrix}.
\end{align}
Therefore,
\begin{align}
     -\frac{1}{2}\Tr[\operatorname{ad}_G^2(\tilde\rho)\tilde H] &= -\frac{1}{2}\Tr(\{L^\dagger L,\rho\}H-2L\rho L^\dagger H) \notag \\
     &=\Tr(\L(\rho) H).
\end{align}
Finally, we show that the first-order Lindbladian optimality is satisfied:
    \begin{align}
        \left.\frac{d}{dt}\Tr(e^{t\L}(\rho) H)\right|_{t=0}&=\Tr(\L(\rho)H) \notag \\
        &= -\frac{1}{2}\Tr(\operatorname{ad}_G^2(\tilde\rho)\tilde H) \notag\\
        &=-\frac{1}{4}\sum_{j,k}\alpha_j\alpha_k\Tr(\{\operatorname{ad}_{P_j},\operatorname{ad}_{P_k}\}(\tilde\rho)\tilde H)\notag \\
        &= \frac{1}{2}\sum_{j,k}\alpha_j\alpha_k K_{ij} \tag{by \Cref{lem: Hessian_form}}\\
        &= \frac{1}{2}\alpha^\dagger K \alpha\geq 0.
    \end{align}
\end{proof}

\end{document}